\documentclass[aps,manuscript,pop,showpacs,preprint,superscriptaddress]{revtex4-1}
\usepackage{lmodern}

\usepackage{lmodern}
\usepackage[T1]{fontenc}
\usepackage[latin9]{inputenc}
\usepackage{amsmath}
\usepackage{amsthm}
\usepackage{mathdots}
\usepackage{graphicx}

\makeatletter
\theoremstyle{plain}
\newtheorem{thm}{\protect\theoremname}
\ifx\proof\undefined
\newenvironment{proof}[1][\protect\proofname]{\par
	\normalfont\topsep6\p@\@plus6\p@\relax
	\trivlist
	\itemindent\parindent
	\item[\hskip\labelsep\scshape #1]\ignorespaces
}{%
	\endtrivlist\@endpefalse
}
\providecommand{\proofname}{Proof}
\fi

\usepackage{amsthm}\usepackage{latexsym}\usepackage{bm}\usepackage{amsfonts}

\makeatother

\providecommand{\theoremname}{Theorem}

\begin{document}
\title{PT-symmetry entails pseudo-Hermiticity regardless of diagonalizability}
\author{Ruili Zhang}
\affiliation{School of Science, Beijing Jiaotong University, Beijing 100044, China}
\author{Hong Qin }
\thanks{Corresponding author, hongqin@princeton.edu}
\affiliation{Plasma Physics Laboratory, Princeton University, Princeton, NewJersey
08543, USA}
\affiliation{School of Physical Sciences, University of Science and Technology
of China, Hefei 230026, China}
\author{Jianyuan Xiao}
\affiliation{School of Physical Sciences, University of Science and Technology
of China, Hefei 230026, China}
\begin{abstract}
We prove that in finite dimensions, a Parity-Time (PT)-symmetric Hamiltonian
is necessarily pseudo-Hermitian regardless of whether it is diagonalizable
or not. This result is different from Mostafazadeh's, which requires
the Hamiltonian to be diagonalizable. PT-symmetry breaking often occurs
at exceptional points where the Hamiltonian is not diagonalizable.
Our result implies that PT-symmetry breaking is equivalent to the
onset of instabilities of pseudo-Hermitian systems, which was systematically
studied by Krein et al. in 1950s. In particular, we show that the
mechanism of PT-symmetry breaking is the resonance between eigenmodes
with different Krein signatures.
\end{abstract}
\maketitle
In quantum physics, observables are assumed to be Hermitian operators.
Bender and collaborators \cite{bender1998real,bender2002complex,bender2007making}
proposed to relax this fundamental assumption and considered Parity-Time
(PT)-symmetric operators. The concept and techniques of PT-symmetry
have been applied to many branches of physics \cite{jones1999energy,dorey2007ode,makris2008beam,klaiman2008visualization,longhi2009bloch,schomerus2010quantum,chong2011p,feng2011nonreciprocal,szameit2011p,schindler2011experimental,regensburger2012parity,peng2014parity,ablowitz2016inverse,jahromi2017statistical,hodaei2017enhanced,qin2019kelvin}.
Although PT-symmetry was first studied in infinite-dimensional systems,
many of the current applications are in finite dimensions.

When discussing PT-symmetry, a related property, pseudo-Hermiticity,
is often considered. Pseudo-Hermitian operators were introduced by
Dirac and Pauli as a class of non-Hermitian operators \cite{Dirac1942,Pauli1943,Lee69}.
Investigating the relation between PT-symmetry and pseudo-Hermiticity
may reveal important mathematical and physical structures of non-Hermitian
operators. In this regard, Mostafazadeh proved that a diagonalizable
PT-symmetric Hamiltonian is pseudo-Hermitian \cite{mostafazadeh2002pseudo,mostafazadeh2002pseudoII,mostafazadeh2002pseudoIII}.

In this paper, we prove that in finite dimensions, a PT-symmetric
Hamiltonian is necessarily pseudo-Hermitian regardless of whether
it is diagonalizable or not. We first prove that for a Hamiltonian
$H$, a sufficient and necessary condition of pseudo-Hermiticity is
that $H$ is similar to its Hermitian conjugate $\overline{H}$ (Theorem
\ref{thm:2}). Then because a PT-symmetric Hamiltonian is similar
to its Hermitian conjugate, it is pseudo-Hermitian (Theorem \ref{thm:3}).
We emphasize that this result is different from Mostafazadeh's \cite{mostafazadeh2002pseudo,mostafazadeh2002pseudoII,mostafazadeh2002pseudoIII}.
The difference is significant, because our result relaxes the diagonalizability
requirement. As we know, most of the interesting PT-symmetry breaking
happens at exceptional points where the Hamiltonian is not diagonalizable.
Our result is applicable when studying these effects.

As such an application, we show that a theoretical description of
PT-symmetry breaking, which is arguably the most important topic in
PT-symmetry physics, can be built upon the mathematical work on the
instabilities of pseudo-Hermitian systems developed by Krein, Gel\textquoteright fand
and Lidskii \cite{Krein1950,Gel1955,KGML1958} in 1950s. For a pseudo-Hermitian
Hamiltonian, its eigenvalues are symmetric with respect to the real
axis. As the system parameters vary, a necessary and sufficient condition
for the onset of instability is that two eigenmodes with opposite
Krein signatures collide, which is the so-called Krein collision.
These results can be directly applied to PT-symmetric Hamiltonians
due to Theorem \ref{thm:3}, implying that PT-symmetry breaking occurs
when and only when eigenmodes with different Krein signatures collide.
Note that when PT-symmetry breaking happens, the Hamiltonian can be
either diagonalizable or non-diagonalizable. But PT-symmetry is often
broken at the exceptional points where the Hamiltonian is not diagonalizable.
As an example, we show that the governing equations of the classical
Kelvin-Helmholtz instability, which was proven to be PT-symmetric
\cite{qin2019kelvin}, is pseudo-Hermitian, and the Kelvin-Helmholtz
instability is the result of PT-symmetry breaking triggered by the
Krein collision.

We start from the definitions of PT-symmetry, pseduo-Hermiticity,
and another related concept, i.e., G-Hamiltonian matrix. Consider
the linear system specified by a Hamiltonian $H$,
\begin{equation}
\dot{\boldsymbol{x}}=-iH\boldsymbol{x}=A\boldsymbol{x}\thinspace,\label{eq:ls}
\end{equation}
where $A$ is defined to be a shorthand notation of $-iH.$

The Hamiltonian $H$ in Eq.\,(\ref{eq:ls}) is called PT-symmetric
\cite{bender1998real,bender2002complex,bender2007making} if it commutes
with the parity-time operator $PT$, i.e.,
\begin{equation}
PTH-HPT=0\thinspace.\label{eq:PT}
\end{equation}
Here $P$ is a linear operator satisfying $P^{2}=I$ and $T$ is the
complex conjugate operator. In the present study, we will focus on
finite-dimensional systems, for which $H$, $A$ and $P$ can be represented
by matrices, and Eq.\,\eqref{eq:PT} is equivalent to
\begin{equation}
P\bar{H}-HP=0\,,\label{eq:PT1}
\end{equation}
where $\bar{H}$ denotes the complex conjugates of $H$.

The Hamiltonian $H$ in Eq.\,(\ref{eq:ls}) is called pseudo-Hermitian
\cite{Dirac1942,Pauli1943,Lee69} if there exits a non-singular Hermitian
matrix $G$ such that
\begin{equation}
H^{\dagger}G-GH=0\thinspace,\label{eq:PH}
\end{equation}
where $H^{\dagger}$ is the conjugate transpose of the matrix $H$
.

The matrix $A=-iH$ in Eq.\,(\ref{eq:ls}) is called G-Hamiltonian
\cite{Krein1950,Gel1955,KGML1958} if there exist a non-singular Hermitian
matrix $G$ and a Hermitian matrix $S$ such that

\begin{equation}
A=iG^{-1}S.\label{eq:GH}
\end{equation}

The concept of pseudo-Hermiticity was first introduced by Dirac and
Pauli in 1940s \cite{Dirac1942,Pauli1943,Lee69}. G-Hamiltonian matrix
was defined by Krein et al. in 1950s \cite{Krein1950,Gel1955,KGML1958}
in the study of linear dynamical systems satisfying the G-Hamiltonian
condition (\ref{eq:GH}). For finite-dimensional systems, these two
concepts are equivalent.
\begin{thm}
\label{thm:1} For a finite-dimensional system $\dot{\boldsymbol{x}}=-iH\boldsymbol{x}=A\boldsymbol{x}$,
$H$ is pseudo-Hermitian if and only if $A$ is a G-Hamiltonian matrix.
\end{thm}
The proof of Theorem \ref{thm:1} is straightforward according to
the definitions of pseudo-Hermitian and G-Hamiltonian matrices. But
we give this fact the status of a theorem to highlight the exact equivalence
between these two concepts independently defined by physicists and
mathematicians. We will mostly use the terminology of pseudo-Hermiticity
exclusively hereafter.

Now we establish a necessary and sufficient condition for pseudo-Hermiticity.
\begin{thm}
\label{thm:2}For a matrix $H\in C^{n\times n}$, it is pseudo-Hermitian
if and if only it is similar to its complex conjugate $\bar{H}$.
\end{thm}
\begin{proof}
Necessity is easy to prove. If a Hamiltonian is pseudo-Hermitian,
i.e., satisfying Eq.\,\eqref{eq:PH}, then $H=G^{-1}H^{\dagger}G$.
Thus matrix $H$ is similar to $H^{\dagger}$, and also to $\bar{H}$.

We prove the sufficiency by constructing the Hermitian matrix $G$.
Matrix $H$ can be written as
\begin{equation}
H=Q^{-1}JQ\thinspace,
\end{equation}
where $J$ is its Jordan canonical form and $Q$ is a reversible matrix.
The Jordan canonical form consists of several Jordan blocks of the
form
\begin{equation}
J(\lambda)=\left(\begin{array}{cccc}
\lambda & 1\\
 & \ddots & \ddots\\
 &  & \lambda & 1\\
 &  &  & \lambda
\end{array}\right)_{m\times m}\thinspace.
\end{equation}
When $m=1$, the Jordan block $J(\lambda)$ is reduced to $\lambda$.
If $H$ is similar to $\bar{H}$, then its eigenvalues are symmetric
with respect to the real axis, and they are either real numbers or
complex number pairs of the form $\lambda=a+bi$ and $\bar{\lambda}=a-bi$,
where $a$ and $b$ are real numbers. Accordingly, there are two kinds
of matrix blocks
\begin{equation}
\begin{aligned}F_{1}=J(a)=\left(\begin{array}{cccc}
a & 1\\
 & \ddots & \ddots\\
 &  & a & 1\\
 &  &  & a
\end{array}\right)_{m\times m}\thinspace\text{and } & F_{2}=\left(\begin{array}{cc}
J(a+bi) & 0\\
0 & J(a-bi)
\end{array}\right)_{2l\times2l}.\end{aligned}
\label{eq:J123}
\end{equation}
The Jordan matrix can now be expressed as $J=Diag(M_{1},M_{2},\cdots,M_{k}),$
where $M_{j}$ is in the form of $F_{1}$ or $F_{2}$. In the following,
we prove that both types of matrix blocks are pseudo-Hermitian. For
both types of matrix blocks, we find that Hermitian matrix
\begin{equation}
G_{j}^{'}=\left(\begin{array}{cccc}
0 & \cdots & 0 & 1\\
\vdots & \iddots & 1 & 0\\
0 & \iddots & \iddots & \vdots\\
1 & 0 & \cdots & 0
\end{array}\right)
\end{equation}
satisfies the condition $M_{j}^{\dagger}G_{j}^{'}-G_{j}^{'}M_{j}=0$,
i.e., $M_{j}$ is pseudo-Hermitian. Next we construct a larger Hermitian
matrix $G^{'}$ using $G_{j}^{'}$  as follows,
\begin{equation}
G^{'}=Diag(G_{1}^{'},G_{2}^{'},\cdots,G_{k}^{'}),
\end{equation}
and the Jordan canonical form of $H$ satisfies $J^{\dagger}G^{'}-G^{'}J=0$.
Let
\begin{align}
G & =Q^{\dagger}G^{'}Q\thinspace,\label{eq:Q-1}
\end{align}
and we obtain
\begin{equation}
\begin{aligned}H^{\dagger}G-GH\thinspace & =\left(Q^{-1}JQ\right)^{\dagger}G-GQ^{-1}JQ\\
 & =Q^{\dagger}J^{\dagger}Q^{-\dagger}Q^{\dagger}G^{'}Q-Q^{\dagger}G^{'}QQ^{-1}JQ\\
 & =Q^{\dagger}\left(J^{\dagger}G^{'}-G^{'}J\right)Q\\
 & =0,
\end{aligned}
\end{equation}
where $G$ is a non-singular Hermitian matrix. This completes the
proof that $H$ is pseudo-Hermitian.
\end{proof}
The theorem is proved by constructing a non-singular Hermitian matrix
$G$ for the similarity transformation between $H$ and $\bar{H}$.
But $G$ is not unique. For a given $H$, we can find more than one
non-singular Hermitian matrices $G$. In practice, one does not need
to follow the construction procedure given in Theorem \ref{thm:2}
to find $G$. It is often found by direct calculation.
\begin{thm}
For finite-dimensional systems, a PT-symmetric Hamiltonian $H$ is
necessarily pseudo-Hemitian. \label{thm:3}
\end{thm}
\begin{proof}
By the definition of PT-symmetry, i.e., Eq.\,\eqref{eq:PT1}, $H$
is similar to $\bar{H}$. Thus, according to Theorem \ref{thm:2},
it is pseudo-Hermitian.
\end{proof}
Theorem \ref{thm:3} is the main theorem in this paper, and we would
like to emphasize again that it holds regardless of whether $H$ is
diagonalizable or not. We note that Mostafazadeh\textquoteright s
result \cite{mostafazadeh2002pseudo,mostafazadeh2002pseudoII,mostafazadeh2002pseudoIII},
which states that diagonalizable PT-symmetric Hamiltonians are pseudo-Hermitian,
is different from Theorem \ref{thm:3}.

As an application of Theorem \ref{thm:3}, we investigate the mechanism
of PT-symmetry breaking in the framework of pseudo-Hermiticity. Theorem
\ref{thm:3} implies that PT-breaking is equivalent to the onset of
instabilities of pseudo-Hermitian matrices, which was systematically
studied by Krein, Gel'fand and Lidskii \cite{Krein1950,Gel1955,KGML1958}
in 1950s. Specifically, the instability analysis of pseudo-Hermitian
matrices gives a comprehensive description on how real eigenvalues
of $H$ evolve into conjugate pairs of complex eigenvalues as the
system parameters vary. Here we briefly summarize the main results.
(i) The eigenvalues of a pseudo-Hermitian Hamiltonian $H$ are symmetric
with respect to real axis. They are either real numbers or complex
conjugate pairs. (ii) Let $\psi$ be an eigenmode (or eigenvector)
of $H$, Krein product of $\psi$ can be defined as \cite{Krein1950,Gel1955,KGML1958}
\[
\left\langle \mathbf{\psi},\psi\right\rangle =\psi^{\dagger}G\psi\thinspace.
\]
The sign of the Krein product is called Krein signature. It was found
that the physical meaning of the Krein product is action \cite{zhang2016structure},
which is partially indicated by the fact that its dimension is {[}energy{]}$\times${[}time{]}.
We will also refer to the Krein product as action, especially in the
context of physics. (ii) The eigenvalues of $H$ can be classified
according to the Krein products of the corresponding eigenvectors.
An $r$-fold real eigenvalue $\lambda$ of $H$ with its eigen-subspace
$V_{\lambda}$ is called the first kind if all eigenmodes of $\lambda$
have positive actions, i.e., $\left\langle \boldsymbol{y},\boldsymbol{y}\right\rangle >0$
for any $\boldsymbol{y}\neq0$ in $V_{\lambda}$. It is called the
second kind if all eigenmodes of $\lambda$ have negative actions.
If there exists a zero-action eigenmode, then $\lambda$ is called
an eigenvalue of mixed kind \cite{KGML1958}. If an eigenvalue is
the first kind or the second kind, it's called definite. (iii) The
number of each kind of eigenvalues is determined by the Hermitian
matrix $G$. Let $p$ be the number of positive eigenvalues and $q$
be the number of negative eigenvalues of the matrix $G$, then any
pseudo-Hermitan Hamiltonian has $p$ eigenvalues of first kind and
$q$ eigenvalues of second kind (counting multiplicity). (iv) The
finite-dimensional pseudo-Hermitian Hamiltonian is strongly stable
if and only if all of its eigenvalues lie on the real axis and are
definite. Here, a pseudo-Hermitian Hamiltonian is strongly stable
means that eigenvalues of all pseudo-Hermitian Hamiltonians in an
open neighborhood of the parameter space lie on the real axis.  As
a result, a pseudo-Hermitian Hamiltonian becomes unstable when and
only when a positive-action mode resonates with a negative-action
mode. This is a process known as the Krein collision.

Applying these results to PT-symmetric Hamiltonians, we see that PT-symmetry
breaking can happen only when a repeated eigenvalue appears as a result
of two eigenmodes resonate. However, if two eigenmodes with the same
sign of action resonate, then there is no PT-symmetry breaking. PT-symmetry
breaking is triggered only when a positive-action mode resonates with
a negative-action mode.

Let's look at an example. The governing equations for the classical
Kelvin-Helmholtz instability in fluid dynamics was shown to be a complex
system with the following PT-symmetric Hamiltonian \cite{qin2019kelvin}
\begin{equation}
H=\left(\begin{array}{cc}
\dfrac{-k(-u_{10}\rho_{10}-2u_{20}\rho_{20}+u_{10}\rho_{20})}{\rho_{10}+\rho_{20}} & \dfrac{-i|k|(u_{10}-u_{20})^{2}\rho_{20}+ig(\rho_{20}-\rho_{10})}{\rho_{10}+\rho_{20}}\\
-i|k| & ku_{10}
\end{array}\right).
\end{equation}
According to Theorem \ref{thm:3}, it is also a pseudo-Hermitian Hamiltonian
satisfying Eq.\,\eqref{eq:PH}. With straightforward calculation,
we find the following Hermitian matrix
\begin{equation}
G=\left(\begin{array}{cc}
-|k| & 0\\
0 & \dfrac{|k|(u_{10}-u_{20})^{2}\rho_{20}-g(\rho_{20}-\rho_{10})}{\rho_{10}+\rho_{20}}
\end{array}\right)
\end{equation}
such that $H^{\dagger}G-GH=0$. The eigenvalues of $H$ are
\begin{equation}
\begin{aligned}a_{1}= & \dfrac{k(\rho_{10}u_{10}+\rho_{20}u_{20})-\sqrt{\Delta}}{\rho_{10}+\rho_{20}},\\
a_{2}= & \dfrac{k(\rho_{10}u_{10}+\rho_{20}u_{20})+\sqrt{\Delta}}{\rho_{10}+\rho_{20}},
\end{aligned}
\end{equation}
and the corresponding eigenvectors are
\begin{equation}
\begin{aligned}\phi_{1}= & (\dfrac{-ik\rho_{20}(u_{10}-u_{20})+\sqrt{\Delta}}{|k|(\rho_{10}+\rho_{20})},1),\\
\phi_{2}= & (\dfrac{-ik\rho_{20}(u_{10}-u_{20})-\sqrt{\Delta}}{|k|(\rho_{10}+\rho_{20})},1),
\end{aligned}
\end{equation}
where $\Delta=-|k|g(\rho_{10}^{2}-\rho_{20}^{2})-k^{2}\rho_{10}\rho_{20}(u_{10}-u_{20})^{2}$
. The Krein signatures, or the signs of actions, of the eigenvalues
of $H$ can be determined by the Hermitian matrix $G$. When
\[
\tau\equiv\dfrac{|k|(u_{10}-u_{20})^{2}\rho_{20}-g(\rho_{20}-\rho_{10})}{\rho_{10}+\rho_{20}}<0\,,
\]
both eigenvalues of $G$ are negative and the PT-symmetric Hamiltonian
$H=iA$ is stable. When $\tau>0$, one of the eigenvalues of $G$
is positive and the other one is negative. Thus one eigenvalue of
$H$ have a positive action and the other one has a negative action,
and the resonance between them will result in PT-symmetry breaking.
Let's use a numerically calculated examples to observe the breaking
of PT-symmetry. We plot the process in Fig.\,\ref{f1} by fixing
$u_{10}=1$, $\rho_{10}=2,\thinspace\rho_{20}=3$, $k=1$ and $g=3$,
and varying $u_{20}$ from $2.3$ to $2.7$. When $u_{20}=2.3$, the
eigenvalues of $H$ are all real numbers, one of which has a positive
action (marked by $M_{+}$) and the other one has a negative action
(marked by $M_{-}$) in Fig.\,\ref{f1}(a). Fig.\,\ref{f1}(b) shows
that as $u_{20}$ increases, $M_{+}$ and $M_{-}$ move towards each
other. Increasing $u_{20}$ to $\sqrt{5/2}+1=2.58114$, eigenmodes
$M_{+}$ and $M_{-}$ collide on the real axis, as shown in Fig.\,\ref{f1}(c).
Because the resonance is between modes with different sign of actions,
the eigenvalues of $H$ split into a pair symmetric with respect to
the real axis and the PT-symmetry is broken. Fig.\,\ref{f1}(d) shows
that the two eigenvalues of $H$ move out of real axis when $u_{20}=2.7$.

\begin{figure}
\includegraphics[scale=0.8]{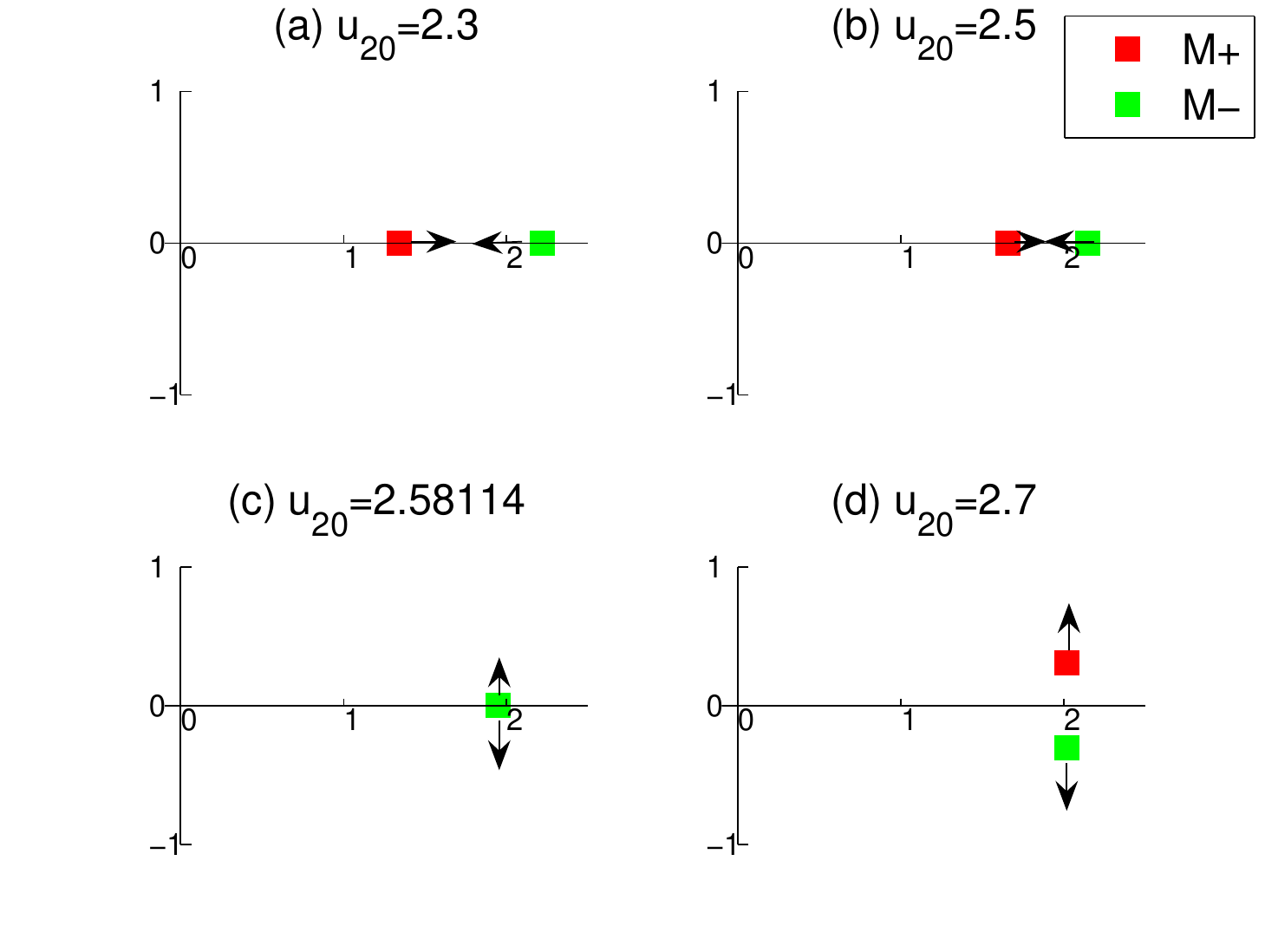}

\caption{PT-symmetry breaking occurs when a positive-action eigenmode (red)
resonates with a negative-action eigenmode (green) .}
\label{f1}
\end{figure}

In summary, we have proved that for finite-dimensional systems, a
PT-symmetric Hamiltonian is necessarily pseudo-Hermitian regardless
of whether it is diagonalizable or not. This result is stronger than
Mostafazadeh's \cite{mostafazadeh2002pseudo,mostafazadeh2002pseudoII,mostafazadeh2002pseudoIII},
which requires that the Hamiltonian is diagonalizable. As we know,
PT-symmetry breaking often happens at exceptional points where the
Hamiltonian is not diagonalizable. The fact that a PT-symmetric Hamiltonian
is always pseudo-Hermitian implies that PT-symmetry breaking is equivalent
to to the onset of instabilities of pseudo-Hermitian matrices. Therefore,
the systematic results by Krein et al. on how a pseudo-Hermitian system
becomes unstable \cite{Krein1950,Gel1955,KGML1958} can be directly
applied to the process of PT-symmetry breaking. In particular, we
showed that PT-symmetry breaking is triggered when and only when two
eigenmodes with different signs of actions resonate. This process
is illustrated using the example of the classical Kelvin-Helmholtz
instability.

We finish our discussion with an observation. Theorem \ref{thm:3}
asserts that a PT-symmetric matrix is necessarily pseudo-Hermitian.
One wonders whether the reverse is true. If the $P$ operator in the
definition of PT-symmetry (\ref{eq:PT}) is not required to be a parity
transformation, i.e., $P^{2}=I$, then a pseudo-Hermitian matrix is
also PT-symmetric according to Theorem \ref{thm:2}. In this case,
PT-symmetry and pseudo-Hermition are equivalent, at least in finite
dimensions. We note that essentially all the spectrum properties associated
with PT-symmetry are still valid when the requirement of $P^{2}=I$
is removed.
\begin{acknowledgments}
This research was supported by the National Natural Science Foundation
of China (NSFC-11775219 and NSFC-11575186), the Fundamental Research
Funds for the Central Universities (Grant No. 2017RC033), China Postdoctoral
Science Foundation (2017LH002), the National Key Research and Development
Program (2016YFA0400600, 2016YFA0400601, 2016YFA0400602 and 2017YFE0301700),
and the U.S. Department of Energy (DE-AC02-09CH11466).
\end{acknowledgments}
%
%

\end{document}